\documentclass[12pt,superscriptaddress]{iopart}
\usepackage{iopams}
\usepackage[dvipdfmx]{graphicx}
\usepackage{epstopdf}

\expandafter\let\csname equation*\endcsname\relax
\expandafter\let\csname endequation*\endcsname\relax
\usepackage{amsmath}
\usepackage{amssymb}
\usepackage{mathrsfs}
\usepackage{amsthm}
\usepackage{bm}

\usepackage{multirow,url}

\newtheorem{proposition}{Proposition} 
\newtheorem{theorem}{Theorem}
\newtheorem{corollary}{Corollary}
\newtheorem{lemma}{Lemma}

\def\cH{\mathcal{H}}

\def\cT{\mathcal{T}}
\def\C{\mathbb{C}}

\def\Pr{\mathbb{P}}

\DeclareMathOperator{\id}{id}

\def\Label#1{\label{#1}\ [\ \text{#1}\ ]\ }
\def\Label{\label}

\usepackage[dvipdfmx]{color} %, hyperref
\usepackage{cite}
\usepackage{braket}

\theoremstyle{definition}

\newtheorem*{example*}{Example}
\newtheorem{remark}{Remark}

\newcommand{\ketbra}[2]{\ket{#1}\!\bra{#2}}

\DeclareMathOperator{\Ker}{Ker}

\begin{document}

\title[Hidden Markovian Process with Quantum Hidden System]{Asymptotic and Non-Asymptotic Analysis for Hidden Markovian Process with Quantum Hidden System}

\author{Masahito Hayashi$^{1,2,3}$ and Yuuya Yoshida$^{1}$}
\address{$^1$ Graduate School of Mathematics, Nagoya University, 
	Nagoya, 464-8602, Japan}
\address{$^2$ Shenzhen Institute for Quantum Science and Engineering, Southern University of Science and Technolog, Nanshan District, Shenzhen, 518055, China}
\address{$^3$ Centre for Quantum Technologies, National University of Singapore, 
	3 Science Drive 2, 117542, Singapore}

\ead{masahito@math.nagoya-u.ac.jp and m17043e@math.nagoya-u.ac.jp}

\bigskip
\begin{indented}
\item 14 July 2018
\end{indented}

\begin{abstract}
We focus on a data sequence produced by 
repetitive quantum measurement on an internal hidden quantum system,
and call it a hidden Markovian process.
Using a quantum version of the Perron-Frobenius theorem, we derive novel upper and lower bounds
for the cumulant generating function of the sample mean of the data.
Using these bounds, we derive the central limit theorem and large and moderate deviations for the tail probability.
Then, we give the asymptotic variance is given by using the second derivative of the cumulant generating function.
We also derive another expression for the asymptotic variance
by considering the quantum version of the fundamental matrix.
Further, we explain how to extend our results to a general probabilistic system.
\end{abstract}

% Uncomment for PACS numbers
%\pacs{00.00, 20.00, 42.10}
%
% Uncomment for keywords
%\vspace{2pc}
\noindent{\it Keywords}: 
quantum system, hidden Markov, central limit theorem, large deviation, moderate deviation, asymptotic variance

\section{Introduction}
Consider a physical system with an internal quantum system.
Usually, it is not so easy to observe the internal system, directly.
When the physical system has a classical output,
we can observe the classical output and other parts cannot be observed.
For example, a quantum random number generator has 
such an internal system and a classical output \cite{HG}.
As another example, such a system appears in a quantum memory of a channel \cite{CGLM,Pascal,Kretchmann}.
Such a correlated system also appears in quantum spin chains \cite{Ogata,BTG}.
This kind of system is formulated to be a hidden Markovian process with quantum hidden system as Fig.~\ref{F1}.
Originally, a classical Markovian process is formulated as
a probability transition matrix.
Then, a classical hidden Markovian process is formulated as
two probability transition matrices,
in which, one describes the Markovian process on the hidden system,
and the other describes the relation between the hidden system and the observed system.

While there are several formulations of quantum analogue of Markovian process,
one natural formulation is a trace-preserving completely positive map (TP-CP map).
However, in this formulation, nobody observes the system, i.e.,
no observation of the quantum system is discussed.
%has no correlation with another system that can be observed.
To introduce a measurability on this system, we need to introduce 
a hidden Markovian process with quantum hidden system.
When the quantum system can be measured, the resultant state depends on the classical output.
That is, the state evolution depends on the classical output $\omega \in \Omega$,
and is described by a set of CP maps $C_\omega$, which is often called
an instrument \cite{Ozawa}.
In this case, the sum $\sum_{\omega \in \Omega} C_\omega$ needs to be trace-preserving.
When the initial state is $\rho$, the classical output $\omega$ is observed
with the probability $\Tr C_\omega(\rho)$ so that the resultant state 
is $C_\omega(\rho)/\Tr C_\omega(\rho) $.
Such a system is initially formulated as a quantum measuring process \cite{Ozawa}
and the set $\{C_\omega\}_{\omega \in \Omega}$ is called an instrument.

In the classical case, there are so many studies for Markovian process.
These studies focus on the random variables $X_i$ generated subject to this process
and consider the sample mean $X^n:= (1/n)\sum_{i=1}^n X_i$.
Similar to the independently and identically distributed case,
the sample mean $X^n$ converges the expectation in probability.
Also, the central limit theorem holds for the sample mean \cite{CLT2,CLT3,CLT4,Ben-Ari}.
Further, the large and moderate deviations also hold \cite{R2} \cite[Theorem~3.1.2]{DZ} \cite[Corollaries 8.3 and 8.4]{W-H}.
However, the non-asymptotic analysis has not been discussed sufficiently.
While in the non-asymptotic analysis, we derive upper and lower bounds for the tail probability,
we need to consider requirements for a good bound because we need to distinguish good bounds from trivial bounds.
Similarly to \cite{Markov-HW}, we impose the following requirements on good bounds.
\begin{description}
\item[Computational complexity] 
In order that the bound works efficiently, we need to calculate the bound efficiently.
For this aim, we need to clarify the computational complexity to calculate the bound,
and the complexity needs to be polynomial with respect to the number $n$ of observation, at least.
\item[Asymptotic tightness]
The bound needs to achieve the optimality in one of the following regimes.
\begin{description}
\item[C1] Large deviation
\item[C2] Moderate deviation
\item[C3] Central limit theorem
\end{description}
\end{description}
The paper \cite{W-H} derived upper and lower bounds for the tail probability, which 
have the computational complexity $O(1)$
and achieve asymptotic tightness in the sense of C1 and C2.
These upper and lower bounds were derived from 
the evaluation of the cumulant generating function.
Large deviation considers the event that the difference between the sample mean and the expectation is greater than 
a certain threshold when the threshold is a constant.
That is, in the large deviation, the event of our interest is largely deviated.
Then, the event of large deviation has exponentially small probability.
Moderate deviation discusses the event when the threshold is a constant is larger than that in the central limit regime but goes to zero.
That is, in the moderate deviation, the event of our interest is moderately deviated.
%Since moderate deviation is not so familiar in general physics, we explain its meaning.
Since the event of the central limit theorem converges to a constant,
the event of the moderate deviation can be regarded as 
the intermediate situation.

In the quantum setting, 
the paper \cite{H-G} derived the large deviation in a similar setting.
The papers \cite{Ogata,BTG} discussed the large deviation in quantum spin chains.
However, no study derived moderate deviation and upper and lower bounds to satisfy the above requirements.
In addition, the paper \cite{GK} addressed local asymptotic normality
in the context of system identification for quantum Markov chains.
But, it did not discuss the central limit theorem of the sample mean when the hidden system is given as a quantum system.

\begin{figure}
\begin{center}
\scalebox{1}{\includegraphics[scale=0.5]{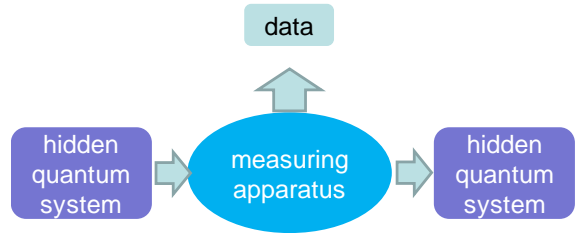}}
\end{center}

\caption{hidden Markovian process with quantum hidden system}
\Label{F1}
\end{figure}%

In this paper, we consider the extension of the upper and lower bounds by \cite{W-H}
to a hidden Markovian process with quantum hidden system.
That is, when a real-valued variable $X_i$ is generated subject to the process,
we focus on the sample mean $X^n:= (1/n)\sum_{i=1}^n X_i$, 
and discuss the asymptotic behavior.
More precisely, we derive the central limit theorem, i.e., we show that
the random variable 
$\sqrt{n}(X^n - E) $ converges to the Gaussian distribution, where $E$ is the expectation value.
Next, we focus on the tail probability of this process.
We derive large and moderate deviations for the tail probability.
In addition, we derive a finite-length evaluation for the tail probability,
i.e., upper and lower bounds of the tail probability
that derives the large and moderate deviations for the tail probability.
In these derivations, we first focus on the quantum version of the Perron-Frobenius theorem \cite{Schrader,Wolf}, 
which characterizes the Perron-Frobenius eigenvalue.
Using the Perron-Frobenius eigenvalue, we derive upper and lower bounds of cumulant generating function of the sample mean, whose computational complexity is $O(1)$.
Then, employing the same method as \cite{W-H}, we show the
asymptotic tightness in the sense of C1 and C2.
% of the sample mean of hidden Markovian process with quantum hidden system.
 Further, we derive the central limit theorem and calculate the asymptotic variance.

The remaining part of this paper is organized as follows.
Sections~\ref{qPF} and \ref{S2} are devoted for mathematical preparations.
Section~\ref{SS5} explains Bregman divergence and its variants, which are powerful tools for our purpose.
Section~\ref{s13} gives the central limit theorem.
Section~\ref{s14} derives upper and lower bounds for the tail probability.
These bounds achieve the tightness in the sense of the large and moderate deviations.
Section~\ref{cal} gives the concrete form of the variance, which appears in 
the moderate deviation and central limit theorem. 
Section~\ref{S8} explains that our model is equivalent to the model of finitely correlated states discussed in \cite{Ogata}.

\section{Perron-Frobenius theorem for quantum systems}\label{qPF}

As a preparation, we summarize basic knowledge for 
a quantum version of the Perron-Frobenius theorem on a finite-dimensional quantum system $\cH$ of interest.
For this aim, we explain the results of \cite{Yoshida}
in the case of completely positive maps. 
First, let us define a few terms used in linear algebra. 
The spectral radius $r(\Lambda)$ of a linear map $\Lambda$ is defined as 
the maximum of the absolute values of all eigenvalues of $\Lambda$ on the linear space $\cT(\cH)$ of all Hermitian matrices on $\cH$ \cite{Wolf}. 
The relations $r(\Lambda_1\otimes\Lambda_2)=r(\Lambda_1)r(\Lambda_2)$ and $r(\Lambda^\ast)=r(\Lambda)$ hold, 
where $\Lambda^\ast$ denotes the adjoint map of a linear map $\Lambda$. 
For a linear map $\Lambda$ and its eigenvalue, 
the multiplicity as a root of the characteristic polynomial is called the algebraic multiplicity, 
and the dimension of the eigenspace is called the geometric multiplicity \cite{Yoshida}. 
The spectral radius plays a special role as follows.
Then, we have the following proposition.

\begin{proposition}[{\cite[Theorem~6]{Yoshida}}]\label{qirr}
	For any completely positive map $\Lambda$, the following conditions are equivalent.
	\renewcommand{\theenumi}{I-\roman{enumi}}
	\begin{enumerate}
		\item\label{mean}
		{The inequality $r(\Lambda)>0$ holds and} 
		there exist strictly positive definite matrices $\rho_0$ and $A_0$ such that 
		$\Tr\rho_0 A_0 = 1$, and any Hermitian matrix $H$ satisfies 
		\begin{equation*}
			\lim_{n\to\infty} \frac{1}{n} \sum_{k=0}^{n-1} (r(\Lambda)^{-1} \Lambda)^k(H) = (\Tr A_0 H)\rho_0.
		\end{equation*}
		\item\label{evec2}
		Both $\Lambda$ and $\Lambda^\ast$ have strictly positive definite eigenvectors associated with $r(\Lambda)>0${, and} 
		the eigenvalue $r(\Lambda)$ of $\Lambda$ has the geometric multiplicity $1$.
		\item\label{exp}
		For any state $\rho>0$, there exists a positive number $t$ such that $e^{t\Lambda}(\rho)>0$.
		\item\label{improve}
		Any state $\rho$ satisfies $(\iota + \Lambda)^{(\dim\mathcal{H})^2 - 1}(\rho)>0$.
		\item\label{ineq}
		If a state $\rho$ and a nonnegative number $\alpha$ satisfy $\Lambda(\rho)\le\alpha\rho$, 
		then $\rho>0$.
		\item\label{evec1}
		$\Lambda$ has no eigenvectors on the boundary of the set of all positive semi-definite matrices.
		\item[(\ref{exp})$'$]
		For any states $\rho$ and $\sigma$ whose ranks equal one, 
		there exists a natural number $n$ such that $\Tr\sigma\Lambda^n(\rho)>0$.
	\end{enumerate}
	Here, for two Hermitian matrices $H$ and $H'$ on $\mathcal{H}$, 
	the relations $H\le H'$ and $H<H'$ mean that $H'-H$ is positive semi-definite and strictly positive definite, respectively. 
	$\iota$ denotes the identity map on $\mathcal{T}(\mathcal{H})$.
\end{proposition}

A $d \times d$ nonnegative matrix $W$ is called irreducible 
when for any $i,j$, there exists a natural number $n$ such that 
{the $(i,j)$ component $(W^n)_{i,j}$ of $W^n$ is positive.}
It can be easily shown that this condition is equivalent to 
the condition that $(e^W)_{i,j}>0$ for any $i,j$.
Hence, Condition (\ref{exp}) can be regarded as quantum extension of the irreducibility.
Hence, a completely positive map $\Lambda$ on $\mathcal{T}(\mathcal{H})$ 
is called irreducible when at least one condition in Proposition \ref{qirr} holds.
The value $r(\Lambda)$ and the matrix $\rho_0$ in Condition (\ref{mean}) are called 
the Perron-Frobenius eigenvalue and a Perron-Frobenius eigenvector of $\Lambda$, respectively. 
It can be shown that the irreducibility of $\Lambda$ implies that of $\Lambda^\ast$, and 
that the two matrices $\rho_0$ and $A_0$ in Condition (\ref{mean}) are eigenvectors of $\Lambda$ and $\Lambda^\ast$, respectively \cite{Yoshida}.
Condition (\ref{exp}) is called ergodicity in \cite{Schrader}.

A irreducible completely positive map $\Lambda$ satisfies Condition (\ref{mean}) 
but $(r(\Lambda)^{-1}\Lambda)^n(\rho)$ might periodically behave as $n\to\infty$. 
Combining Corollary 3 and Theorem~6 of \cite{Yoshida}, we obtain the following conditions that exclude the periodicity.

\begin{proposition}\label{qprim}
	For any completely positive map $\Lambda$, the following conditions are equivalent.
	\renewcommand{\theenumi}{P-\roman{enumi}}
	\begin{enumerate}
		\item\label{p-mean}
		{The inequality $r(\Lambda)>0$ holds and} 
		there exist strictly positive definite matrices $\rho_0$ and $A_0$ such that 
		$\Tr\rho_0 A_0 = 1$, and any Hermitian matrix $H$ satisfies 
		\begin{equation*}
			\lim_{n\to\infty} (r(\Lambda)^{-1} \Lambda)^n(H) = (\Tr A_0 H)\rho_0.
		\end{equation*}
		\item\label{p-evec2}
		Both $\Lambda^{\otimes2}$ and $(\Lambda^{\otimes2})^\ast$ have strictly positive definite eigenvectors associated with $r(\Lambda)^2>0${, and} 
		the eigenvalue $r(\Lambda)^2$ of $\Lambda^{\otimes2}$ has the geometric multiplicity $1$.
		\item\label{p-exp}
		For any state $\rho>0$, there exists a positive number $t$ such that $e^{t\Lambda^{\otimes2}}(\rho)>0$.
		\item\label{p-improve}
		Any state $\rho$ satisfies $(\iota^{\otimes2} + \Lambda^{\otimes2})^{(\dim\mathcal{H})^4 - 1}(\rho)>0$.
		\item\label{p-ineq}
		If a state $\rho$ and a nonnegative number $\alpha$ satisfy $\Lambda^{\otimes2}(\rho)\le\alpha\rho$, 
		then $\rho>0$.
		\item\label{p-evec1}
		$\Lambda^{\otimes2}$ has no eigenvectors on the boundary of the set of all positive semi-definite matrices.
		\item[(\ref{exp})$'$]
		For any states $\rho$ and $\sigma$ whose ranks equal one, 
		there exists a natural number $n$ such that $\Tr\sigma(\Lambda^{\otimes2})^n(\rho)>0$.
	\end{enumerate}
\end{proposition}

A completely positive map $\Lambda$ on $\mathcal{T}(\mathcal{H})$ 
is called primitive when at least one condition in Proposition \ref{qprim} holds. 
It can be shown that the primitivity implies the irreducibility and 
that the primitivity of $\Lambda$ implies that of $\Lambda^\ast$ \cite{Yoshida}.
When a completely positive map maps all diagonal matrices to themselves, 
the above irreducibility and primitivity are equivalent to 
the irreducibility and primitivity in the classical case, respectively. 
As for other equivalent conditions for the primitivity, see \cite[Theorem~6.7]{Wolf}.

Proposition \ref{qprim} has been derived from Corollary 3 and Theorem~6 of \cite{Yoshida}. 
For a completely positive map, Corollary 3 of \cite{Yoshida} is simple as follows.

\begin{proposition}[{\cite[Corollary 3]{Yoshida}}]
	Let $\Lambda$ be a completely positive map. 
	Then, $\Lambda$ is primitive if and only if $\Lambda^{\otimes2}$ is irreducible.
\end{proposition}

The spectral radius of a trace-preserving completely positive map equals one, 
and its adjoint map has the eigenvector $I$ \cite[Section~3.2]{Yoshida}, where $I$ denotes the identity matrix on $\mathcal{H}$. 
Thus, Conditions (\ref{evec2}) and (\ref{p-evec2}) lead to the following corollary immediately.

\begin{corollary}
	Let $\Lambda$ be a trace-preserving completely positive map. 
	Then, $\Lambda$ is irreducible if and only if 
	$\Lambda$ has a fixed state with full rank and the equation $\dim\Ker(\Lambda - \iota)=1$ holds. 
	$\Lambda$ is primitive if and only if $\Lambda$ has a fixed state with full rank 
	and the equation $\dim\Ker(\Lambda^{\otimes2} - \iota^{\otimes2})=1$ holds.
\end{corollary}

Further, we have the following proposition.

\begin{proposition}[{\cite[Corollary 6]{Yoshida}}]\label{MGR}
	Let $\Omega$ be a nonempty finite set, $\Lambda_\omega$ be a completely positive map 
	and $a_\omega$ be a positive number for each $\omega\in\Omega$. 
	If $\Lambda := \sum_\omega \Lambda_\omega$ is irreducible, 
	then so is $\Lambda_a := \sum_\omega a_\omega \Lambda_\omega$. 
	If $\Lambda := \sum_\omega \Lambda_\omega$ is primitive, 
	then so is $\Lambda_a := \sum_\omega a_\omega \Lambda_\omega$.
\end{proposition}

All the statements in this section hold under a more general setting \cite{Yoshida}. 
Sections~\ref{S2}--\ref{s14} derive several properties only with the irreducibility. 
Only Section~\ref{cal} requires the primitivity.

To understand the statements in this section, 
we give a trace-preserving completely positive map that is irreducible but not primitive as follows.

\begin{example*}
	Let $\{ \ket{i} \}_{i=0}^{d-1}$ be an orthonormal basis of $\mathbb{C}^d$ 
	and $\Lambda$ be the trace-preserving completely positive map on $\mathcal{T}(\mathbb{C}^d)$ 
	which maps any Hermitian matrix $H$ to 
	\[
	\Lambda(H) := \sum_{i=0}^{d-1} \braket{i-1 |H| i-1}\ketbra{i}{i},
	\]
	where $i\in\mathbb{Z}/d\mathbb{Z}$. 
	To show that $\Lambda$ is irreducible, we adopt Condition (\ref{ineq}). 
	Let $\rho$ and $\alpha$ be a state and nonnegative number, respectively, satisfying $\Lambda(\rho)\le\alpha\rho$. 
	Suppose $\braket{i|\rho|i}=0$ for an $i$. Then, 
	\[
	0 = \alpha \braket{i|\rho|i} \ge \braket{i|\Lambda(\rho)|i} = \braket{i-1|\rho|i-1},
	\]
	whence $\braket{i-1|\rho|i-1}=0$. 
	Iterating this, we have $\Tr\rho=0$, which contradicts that $\rho$ is a state. 
	Therefore, $\braket{i|\rho|i}>0$ for any $i$. 
	Since the inequality $0<\Lambda(\rho)\le\alpha\rho$ means $\rho>0$, $\Lambda$ is irreducible.
	
	To show that $\Lambda$ is not primitive, we adopt Condition (\ref{p-evec2}). 
	The separable state $(1/d)\sum_{i=0}^{d-1} \ketbra{i}{i}\otimes\ketbra{i}{i}$ is mapped to itself by $\Lambda^{\otimes2}$. 
	The maximally mixed state $(1/d)I$ is an eigenvector associated with $r(\Lambda)=1$ of $\Lambda$. 
	Thus, the two separable states $(1/d)\sum_{i=0}^{d-1} \ketbra{i}{i}\otimes\ketbra{i}{i}$ and $(1/d^2)I^{\otimes2}$ 
	are eigenvectors associated with $r(\Lambda)^2=1$ of $\Lambda^{\otimes2}$. 
	Therefore, Condition (\ref{p-evec2}) does not hold. 
	Moreover, we show that $\Lambda$ does not satisfy Condition (\ref{p-mean}). 
	The sequence $\{\Lambda^n(\ketbra{0}{0})\}_n$ does not converge because $\Lambda(\ketbra{i}{i}) = \ketbra{i+1}{i+1}$ for any $i$. 
	Hence, Condition (\ref{p-mean}) does not hold.\qed
\end{example*}

\section{Evaluations for cumulant generating function}\Label{S2}

We discuss the data sequence of hidden Markovian process 
with quantum hidden system generated from 
quantum measurement process described by
a set of CP maps (an instrument) $C=\{C_\omega\}_{\omega \in \Omega}$ when 
$\Lambda := \sum_{\omega \in \Omega} C_\omega $ is irreducible.
For a symbol $\omega \in \Omega$, we assign it to a real number $x_\omega$.
Since Proposition \ref{MGR} guarantees that
$\Lambda_\theta :=\sum_{\omega \in \Omega} e^{\theta x_\omega} C_\omega$ 
is also an irreducible completely positive map, the completely positive map
$\Lambda_\theta$ has the Perron-Frobenius eigenvalue $\lambda_\theta$
and the Perron-Frobenius eigenvector $\rho_\theta$ whose trace equals one.
We define the function $\phi(\theta)$ to be $\log \lambda_\theta$.
Here, the adjoint map $\Lambda_\theta^*$
has the same Perron-Frobenius eigenvalue $\lambda_\theta$.
{We choose the Perron-Frobenius eigenvector $A_\theta$ of $\Lambda^\ast$ 
such that $A_\theta-I$ is positive semi-definite but not strictly positive definite. 
It is possible because once we take a Perron-Frobenius eigenvector $A$ of $\Lambda_\theta$, 
$A$ is strictly positive definite and 
an appropriate positive number $\alpha$ satisfies that $\alpha A - I$ is positive semi-definite but not strictly positive definite.}
%That is, $A_\theta$ is the Perron-Frobenius eigenvector of $\Lambda^\ast$ such that 
%$A_\theta \ge I$ and $\det (A_\theta-I)=0$.}
%the minimum matrix of all Perron-Frobenius eigenvectors greater than or equal to $I$.

\begin{figure}
\begin{center}
\scalebox{1}{\includegraphics[scale=0.5]{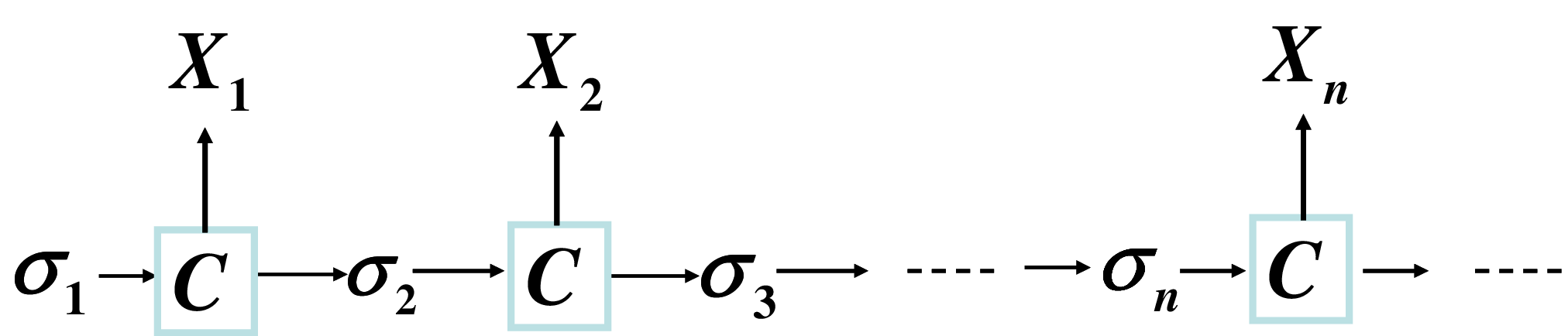}}
\end{center}

\caption{Data sequence from hidden Markovian process with quantum hidden system.
$\sigma_i$ is the state of the hidden quantum system at time $i$.}
\Label{F2}
\end{figure}%

We denote the sequence of observed real numbers by $X_1, \ldots, X_n$ as Fig. \ref{F2}.
When the initial state is $\rho$,
the variables $X_1, \ldots, X_n$ take the values $x_1, \ldots, x_n$, respectively,
with the probability
$\Tr C_{\omega_n} \circ \cdots \circ C_{\omega_1} (\rho) $, where $x_i=x_{\omega_i}$.
In this case, we denote the cumulant generating function
of the variable $n X^n =\sum_{i=1}^n X_i$ by $\phi_{n,\rho}(\theta)$.
When the initial state is the eigenvector $ \rho_\theta$ of $\Lambda_\theta$,
the cumulant generating function $\phi_{n,\rho_\theta}(\theta)$
is calculated as
\begin{align}
& e^{\phi_{n,\rho_\theta}(\theta)}
= \sum_{\omega_1, \ldots, \omega_n}
\Tr e^{\theta x_{\omega_n} }C_{\omega_n} \circ \cdots \circ 
e^{\theta x_{\omega_1}} C_{\omega_1} (\rho_\theta) \nonumber \\
=& \Tr \Lambda_\theta^n (\rho_\theta)
= \lambda_\theta^n \Tr  \rho_\theta = \lambda_\theta^n,\Label{E5-14}
\end{align}
which implies
\begin{align}
\phi_{n,\rho_\theta}(\theta) =n\phi(\theta).
\end{align}

Now, we evaluate the cumulant generating function $\phi_{n,\rho}(\theta)$ for a general initial state $\rho$
by defining
\begin{align}
\overline{\delta}_\rho(\theta):=\log \Tr A_\theta \rho, \quad
\underline{\delta}_\rho(\theta):=\log \Tr A_\theta \rho - \log \|A_\theta\|.
\end{align}
%When $\rho$ is strictly positive definite, we define two real numbers $\delta^+(\rho,\theta)$
Then, we have the following main theorem.
\begin{theorem}\Label{lemma:mgf-finite-evaluation}
The cumulant generating function $\phi_{n,\rho}(\theta)$ of observed variables $X_1, \ldots, X_n$ is evaluated as
\begin{align}
n \phi(\theta)+\underline{\delta}_\rho(\theta)
 \le \phi_{n,\rho}(\theta) \le n \phi(\theta) + \overline{\delta}_\rho(\theta).
\Label{25-12}
\end{align}
\end{theorem}

While this evaluation is very simple,
this evaluation leads so many fruitful derivations for 
asymptotic behavior of hidden Markovian process with quantum hidden system.

\begin{proof}
Since $A_\theta \ge I $, we have
\begin{align}
& e^{\phi_{n,\rho}(\theta)}
= \sum_{\omega_1, \ldots, \omega_n}
\Tr e^{\theta x_{\omega_n} }C_{\omega_n} \circ \cdots \circ 
e^{\theta x_{\omega_1}} C_{\omega_1} (\rho) \nonumber \\
=& \Tr \Lambda_\theta ^n (\rho)
\le \Tr  A_\theta \Lambda_\theta ^n (\rho)
= \Tr  (\Lambda_\theta ^*)^n (A_\theta) \rho \nonumber \\
=& \lambda_\theta^n \Tr  A_\theta \rho
= \lambda_\theta^n e^{\overline{\delta}_\rho(\theta)},
\end{align}
which implies the second inequality in \eqref{25-12}.
Conversely, since $\frac{1}{\|A_\theta\|}A_\theta \le I $, we have
\begin{align}
&e^{\phi_{n,\rho}(\theta)}
=\Tr \Lambda_\theta ^n (\rho)
\ge \Tr \frac{1}{\|A_\theta\|} A_\theta \Lambda_\theta ^n (\rho) \nonumber \\
=& \frac{1}{\|A_\theta\|} \lambda_\theta^n \Tr  A_\theta \rho
= \lambda_\theta^n e^{\underline{\delta}_\rho(\theta)},
\end{align}
which implies the first inequality in \eqref{25-12}.
\end{proof}

\begin{lemma} \Label{L20f}
\begin{align}
\lim_{\theta\to 0} \overline{\delta}_{\rho}(\theta) = 0 ,\quad %\Label{28-4}\\
\lim_{\theta\to 0} \underline{\delta}_{\rho}(\theta) = 0 . \Label{28-5}
\end{align}
\end{lemma}
%%%%%
\begin{proof}
From the construction of $A_\theta$,
$A_\theta$ is continuous for $\theta$.
Hence,
\begin{align}
\lim_{\theta\to 0} 
\Tr  A_\theta \rho
= \Tr A_0 \rho
= \Tr I \rho=1, \Label{28-6}
\end{align}
which implies the first equation of (\ref{28-5}).
Similarly,
\begin{align}
\lim_{\theta\to 0} 
\|A_\theta\|
=
\|I \|=1.\Label{28-7}
\end{align}
Combining (\ref{28-6}) and (\ref{28-7}),
we obtain the second equation of (\ref{28-5}).
\end{proof}

By taking the limit in \eqref{25-12} of Lemma~\ref{lemma:mgf-finite-evaluation},
we have the following.
\begin{corollary}\label{CC2}
%\cite[Theorem~3.1.1]{DZ} 
\Label{theorem:mgf-large-deviation-single-letter}
For $\theta \in \mathbb{R}$, we have
\begin{equation}
\lim_{n \to \infty} \frac{1}{n} \phi_{n,\rho}(\theta) = \phi(\theta).
\end{equation}
\end{corollary}

In fact, $\lambda_\theta$ is defined as the minimum solution of the eigenequation of $\Lambda_\theta $.
Hence, the implicit theorem guarantees that
$\lambda_\theta$ is $C^1$-continuous.
Hence, we find that $\phi$ is $C^1$-continuous.
Since $\phi_{n,\rho}$ is convex, the limit $\lim_{n \to \infty} (1/n)\phi_{n,\rho}(\theta)$ is also convex, i.e., $\phi$ is convex.
Therefore, $\frac{d \phi}{d\theta}$ is a continuously increasing function.

\if0
?????
\begin{corollary}
\Label{C-2}
For $\theta \in \mathbb{R}$, we have
\begin{equation}
\lim_{n \to \infty} 
\frac{1}{n} \frac{d \phi_n}{d \theta}(0) 
= \phi'(0).
\end{equation}
\end{corollary}

\begin{proof}
\begin{equation}
\lim_{n \to \infty} 
\frac{1}{n} \frac{d \phi_n}{d \theta}(0) 
=
\lim_{\theta \to 0}
\frac{\phi_n(\theta)}{\theta}
= \phi'(0).
\end{equation}
\end{proof}
\fi

\section{Bregman divergence and its variants}\label{SS5}
Since the cumulant generating function $\phi_n(\theta)$ is convex,
this corollary implies that $\phi(\theta)$ is convex.
Since the eigenvalue is given as the solution of 
the eigenequation of the linear map $\Lambda_\theta$, 
and $\Lambda_\theta$ is differentiability with respect to $\theta$,
the maximum eigenvalue $\lambda_\theta$ and the eigenvector $\rho_\theta$ 
are differentiable with respect to $\theta$.
%In the following, we assume that $\phi(\theta)$
Therefore, $\phi'(\theta):=\frac{d\phi}{d\theta}(\theta)$
is monotonically increasing with respect to $\theta$.
So, we can define the inverse function ${\phi'}^{-1}$.
For the latter discussion, we introduce 
Bregman divergence $D(\theta\|\bar{\theta})$ \cite{Br}
and
R\'{e}nyi type of Bregman divergence $D_{1+s}(\theta\|\bar{\theta})$ 
\cite[(4.16)]{W-H2}
for the convex function $\phi$ as
\begin{align}
D(\theta\|\bar{\theta}):= &
(\theta-{\bar{\theta}})\phi'(\theta)- \phi(\theta)+ \phi(\bar{\theta}) \Label{1-1},\\
D_{1+s}({\theta} \| {\bar{\theta}}):=&
\frac{\phi((1+s)\theta -s \bar{\theta})-(1+s) \phi(\theta) +s \phi(\bar{\theta})}{s}.
 \Label{1-2}
\end{align}

Since the relation
\begin{align*}
&\frac{d D_{1+s}({\theta} \| {\bar{\theta}})}{ds} \\
=& \frac{
\phi(\theta)
-\phi( \theta + s (\theta-\bar{\theta}))
-s(\bar{\theta}-\theta) \phi'( \theta + s (\theta-\bar{\theta})) }{s^2} \\
=&
 \frac{(1/2)\phi''( \theta + s (\theta-\bar{\theta}))
(\xi s (\bar{\theta}-\theta))^2 }{s^2} >0
 \end{align*}
holds with some parameter $\xi \in (0,1)$,
R\'{e}nyi type of Bregman divergence $D_{1+s}({\theta} \| {\bar{\theta}}) $
is monotonically increasing with respect to $s$.
Also, R\'{e}nyi type of Bregman divergence recovers Bregman divergence with the limit $s \to 0$ as
\begin{align}
\lim_{s \to 0}D_{1+s}({\theta} \| {\bar{\theta}})=
D(\theta\|\bar{\theta}).
\end{align}

Then, %we have the following lemma.
the properties of a differentiable convex function lead the following lemma.
\begin{lemma}\Label{L23}
When $a > \phi'(0)$,
\begin{align}
\begin{split}
%& \inf_{s > 0 \atop \theta > {\phi'}^{-1}(a)} D_{1+s}({\theta} \| {0} )\\
& \inf_{s > 0 \atop \theta > {\phi'}^{-1}(a)} \frac{\phi((1+s)\theta) - (1+s)\phi(\theta)}{s} \\
=& {\phi'}^{-1}(a) a - \phi({\phi'}^{-1}(a)) 
= \sup_{\theta \ge 0}[ \theta a - \phi(\theta) ] \\
=& D({\phi'}^{-1}(a) \| 0 ) .
\end{split}
\Label{26-1}
\end{align}
Similarly, when $a < \phi'(0)$,
\begin{align*}
&\inf_{s > 0 \atop \theta < {\phi'}^{-1}(a)} \frac{\phi((1+s)\theta) - (1+s)\phi(\theta)}{s} \\
=& {\phi'}^{-1}(a) a - \phi({\phi'}^{-1}(a)) 
= \sup_{\theta \le 0}[ \theta a - \phi(\theta) ] \\
=&D({\phi'}^{-1}(a) \| 0 ) .
\end{align*}
\end{lemma}

\section{Central limit theorem}\Label{s13}
Next, we discuss the central limit theorem for the sample mean $X^n$.
Since 
$\Tr \Lambda_0^n ( \frac{d\rho_\theta}{d\theta}|_{\theta=0})
=\Tr \frac{d\rho_\theta}{d\theta}|_{\theta=0}
=0$, 
taking the derivative with respect to $\theta$ at $\theta=0$ in \eqref{E5-14}, 
we have
\begin{align}
&n \sum_{\omega} x_{\omega} \Tr C_\omega (\rho_0)
=
\Tr \frac{d \Lambda_\theta^n}{d \theta}|_{\theta=0} (\rho_0)
+\Tr \Lambda_0^n ( \frac{d\rho_\theta}{d\theta}|_{\theta=0}) \nonumber \\
=& \frac{d \Tr \Lambda_\theta^n (\rho_\theta)}{d \theta}|_{\theta=0}
= \frac{d \lambda_\theta^n}{d \theta}|_{\theta=0}
= n\phi'(0).
\end{align}
That is, the derivative $\phi'(0)$ expresses the expectation 
of $X$ when the initial state is the stationary state $\rho_0$.
That is, to calculate the derivative $\phi'(0)$, it is enough to find the eigenvector $\rho_0$
of $\Lambda_0$ and calculate the expectation under the state $\rho_0$.
Even when the initial state is not the stationary state $\rho_0$,
we can see that the sample mean $X^n$ converges to the derivative $\phi'(0)$
in probability as follows.

%Assume that $\phi(\theta)$ is twice differentialble at $\theta=0$.
Using Theorem~\ref{lemma:mgf-finite-evaluation},
we can characterize the cumulant generating function of 
the random variable $\sqrt{n}
( X^n - \phi'(0) )$ as follows.

\begin{theorem}%[\cite{Lalley}] 
\Label{L6}
The cumulant generating function of 
the random variable 
$\sqrt{n}
( X^n - \phi'(0) )$ 
converges as follows.
%We have
%\lim_{\delta \to \infty} \frac{1}{\delta^2} 
\begin{align}
%\lim_{n \to \infty}
& \log \mathsf{E} \Big[
\exp \Big[\delta \sqrt{n} \Big(
X^n - \phi'(0)
\Big) \Big] \Big] \nonumber \\
=&
%\lim_{n \to \infty}
 \phi_n\Bigl( \frac{\delta}{\sqrt{n}} \Bigr)
-\delta\sqrt{n} 
\phi'(0)
\to \frac{\delta^2}{2}\phi'' (0). 
\Label{28-9}
\end{align}
\end{theorem}

%%%
\begin{proof}
Using (\ref{25-12}) and (\ref{28-5}), we have
\begin{align*}
& \lim_{n \to \infty}
 \phi_n\Bigl( \frac{\delta}{\sqrt{n}} \Bigr)
-\delta\sqrt{n} \phi'(0)
 \nonumber \\
\le &
\lim_{n \to \infty}
n \phi\Bigl( \frac{\delta}{\sqrt{n}} \Bigr)
-\delta\sqrt{n} \phi'(0)
+ \overline{\delta}_{\rho}\Bigl( \frac{\delta}{\sqrt{n}} \Bigr)
 \nonumber \\
=& 
\lim_{n \to \infty}
\delta^2 \frac{\phi\bigl( \frac{\delta}{\sqrt{n}} \bigr) - \bigl( \frac{\delta}{\sqrt{n}} \bigr) 
\phi' (0)
}{\bigl( \frac{\delta}{\sqrt{n}} \bigr)^2} 
= \frac{\delta^2}{2}\phi''(0).
\end{align*}
Similarly, the opposite inequality can be shown by (\ref{25-12}) and (\ref{28-5}).
Hence, we obtain the desired relation.
\end{proof}

The right hand side of (\ref{28-9}) is the cumulant generating function of Gaussian distribution with the variance $\phi''(0)$ and average $0$.
Since the limit of the cumulant generating functions uniquely decides the limit of distribution functions \cite{CRRao},
Lemma~\ref{L6} reproduces the central limit theorem 
as a corollary.

\begin{corollary}\Label{Co1} %\cite{CLT2,CLT3,CLT4,Lalley,Ben-Ari}
The sample mean $X^n$ converges to $\phi'(0)$ in probability.
The limiting distribution of 
$\sqrt{n}(X^n - \phi'(0))$ 
is characterized as
\begin{align}
\lim_{n \to \infty} 
\Pr\{ \sqrt{n}(X^n - \phi'(0))\le\delta \} = 
\Phi\Bigl( \frac{\delta}{\sqrt{\phi''(0)}} \Bigr),
%\lim_{n\to\infty } \Pr 
%\{ \sqrt{n}(\frac{g^n(X^{n+1})}{n}- \eta(0)) \le \delta \}
%= \Phi(\frac{\delta}{\sqrt{\phi'(0)}}).
\end{align}
where
$\Phi(y):= \int_{-\infty}^{y} 
\frac{e^{-x^2/2}}{\sqrt{2\pi}} d x$.
\end{corollary}

\section{Tail probability}\Label{s14}
Next, we proceed to the evaluations for the tail probability.
Using the notations ${\phi'}^{-1}(a)$, 
$\underline{\delta}_{\rho}(\theta)$,
$\overline{\delta}_{\rho}(\theta)$,
and
$\phi(\theta)$ defined in Section~\ref{S2}
and (\ref{25-12}) of Theorem~\ref{lemma:mgf-finite-evaluation}, 
we can derive the following lower bound on the exponent.

\begin{theorem} \Label{P4}
For any $a > \phi'(0)$, we have
\begin{align}
\begin{split}
& - \log \Pr \{ X^n \ge a \} \ge 
\sup_{\theta \ge 0}[ n \theta a - n \phi(\theta) - \overline{\delta}_{\rho}(\theta)] \\
\ge& n{\phi'}^{-1}(a) a - n\phi({\phi'}^{-1}(a)) - \overline{\delta}_{\rho}({\phi'}^{-1}(a)). 
\end{split}
\Label{27-32}
\end{align}
Similarly, for $a <\phi'(0)$, we have
\begin{align*}
&- \log \Pr \{ X^n \le a \}
\ge
\sup_{\theta \le 0}[ n \theta a - n \phi(\theta)- \overline{\delta}_{\rho}(\theta)]  \\
\ge& n {\phi'}^{-1}(a) a - n \phi({\phi'}^{-1}(a)) - \overline{\delta}_{\rho}({\phi'}^{-1}(a)).
\end{align*}
\end{theorem}

\begin{proof}
These evaluations follow from Proposition A.1 and Lemma~4.1 of \cite{W-H}
and (\ref{25-12}) of Theorem~\ref{lemma:mgf-finite-evaluation}.
The proof is the same as Theorem~8.1 of \cite{W-H}.
\end{proof}

Further, using Theorem~\ref{lemma:mgf-finite-evaluation},
we can derive bounds of the opposite direction as follows.
%%%%
\begin{theorem} \Label{T4}
For any $a > \phi'(0)$, we have
\begin{align}
\begin{split}
\lefteqn{ - \log \Pr
\{ X^n \ge a \} 
} 
 \\
%\stackrel{(a)}{\le} 
\le & 
 \inf_{s > 0 \atop \theta \in \mathbb{R}, {\bar{\theta}} \le 0}
 \frac{1}{s}\Bigl[ n \phi((1+s)\theta) - n (1+s)\phi(\theta) 
% \\
 + \overline{\delta}_{\rho}((1+s)\theta) - \underline{\delta}_{\rho}(\theta)
 \\ 
&- (1+s) \log \Bigl(1- e^{ - n [{\bar{\theta}} a - \phi(\theta+{\bar{\theta}}) + \phi(\theta) 
+\overline{\delta}_{\rho}(\theta+{\bar{\theta}}) -\underline{\delta}_{\rho}(\theta)] } 
\Bigr) 
\Bigr]
\\
%\lefteqn{ - \log W_0^{n-1}\times P_0 \{ g^n \ge (n-1) a \}  } \nonumber \\
%\stackrel{(b)}{\le} 
\le & 
 \inf_{s > 0 \atop \theta > {\phi'}^{-1}(a)} \frac{1}{s}\Bigl[
n \phi((1+s)\theta) - n(1+s)\phi(\theta) %\\
+ \overline{\delta}_{\rho}((1+s)\theta) - \underline{\delta}_{\rho}(\theta)
\\
& - (1+s) \log \Bigl(1 %\\
- e^{n
[(\theta - {\phi'}^{-1}(a)) a + \phi({\phi'}^{-1}(a)) - \phi(\theta)
+\overline{\delta}_{\rho}({\phi'}^{-1}(a)) -\underline{\delta}_{\rho}(\theta)
]} \Bigr)  
\Bigr]  \\
%\stackrel{(c)}{=} 
=& \inf_{s > 0 \atop \theta > {\phi'}^{-1}(a)} 
n D_{1+s}(\theta\|0)
+\frac{1}{s}[\overline{\delta}_{\rho}((1+s)\theta) - \underline{\delta}_{\rho}(\theta)] \\
 &\quad -\frac{1+s}{s}
 \log \Bigl(1- e^{-n D({\phi'}^{-1}(a)\|{\theta})
+\overline{\delta}_{\rho}({\phi'}^{-1}(a)) -\underline{\delta}_{\rho}(\theta)
} \Bigr).
\end{split}
   \Label{27-33}
\end{align}
Similarly, for any $a < \phi'(0)$, we  have
\begin{align}
\lefteqn{ - \log \Pr 
\{ X^n \le a \}
} \nonumber \\
\le& \inf_{s > 0 \atop \theta \in \mathbb{R}, {\bar{\theta}} \ge 0}
\frac{1}{s} \Bigl[ n\phi((1+s)\theta) - n(1+s)\phi(\theta) 
%\nonumber \\
+ \overline{\delta}_{\rho}((1+s)\theta) - \underline{\delta}_{\rho}(\theta)
\nonumber \\
&
 - (1+s) \log \Bigl(1
% \nonumber \\
- e^{ - n [{\bar{\theta}} a - \phi(\theta+{\bar{\theta}}) + \phi(\theta) 
+\overline{\delta}_{\rho}(\theta+{\bar{\theta}}) -\underline{\delta}_{\rho}(\theta)] } \Bigr) 
\Bigr] 
\nonumber\\
\le& \inf_{s > 0 \atop \theta < {\phi'}^{-1}(a)} 
\frac{1}{s} \Bigl[
 n \phi((1+s)\theta) - (n-1)(1+s)\phi(\theta) 
%\nonumber \\
+ \overline{\delta}_{\rho}((1+s)\theta) - \underline{\delta}_{\rho}(\theta)
\nonumber \\
& - (1+s) \log \Bigl(1
%\nonumber \\
- e^{n [(\theta - {\phi'}^{-1}(a)) a + \phi({\phi'}^{-1}(a)) - \phi(\theta)
+\overline{\delta}_{\rho}({\phi'}^{-1}(a)) -\underline{\delta}_{\rho}(\theta)
]} \Bigr) 
\Bigr] 
\nonumber\\
 =& \inf_{s > 0 \atop \theta < {\phi'}^{-1}(a)} 
n D_{1+s}(\theta\|0)
+\frac{1}{s}[\overline{\delta}_{\rho}((1+s)\theta) - \underline{\delta}_{\rho}(\theta)] \nonumber \\
 &\quad -\frac{1+s}{s}
 \log \Bigl(1- e^{-n D({{\phi'}^{-1}(a)}\|{\theta})
+\overline{\delta}_{\rho}({\phi'}^{-1}(a)) -\underline{\delta}_{\rho}(\theta)
} \Bigr) . \nonumber
\end{align}
\end{theorem}
%%%%

\begin{proof}
The proof can be shown from Theorem~A.2 of \cite{W-H}
in the same way as Theorem~8.2 of \cite{W-H}.
\end{proof}

%%%%%%%%%%%%%%%%%%%%%%%%%%%%%%%%%%%%%%%%%%%%%%%%%%%%%
%%%%%%%%%%%%%%%% Large Deviation %%%%%%%%%%%%%%%%%%%%%%%%%%%%%
%\subsection{Large Deviation} \Label{\thetaubsection:general-markov-ldp}
The computational complexity does not depend on the number $n$ of observation
in the above upper and lower bounds in Theorems~\ref{P4} and \ref{T4}.
Hence, 
%these bounds have computational complexity $O(1)$. Due to the expressions in Theorems~\ref{P4} and \ref{T4},
the above upper and lower bounds are $O(1)$-computable.
These also attain asymptotic tightness in
the large and moderate deviations regimes as shown in the following corollaries
although the large and moderate deviations regimes immediate 
from the combination of G{\"a}rtner-Ellis theorem \cite{Gartner} and Corollary \ref{CC2}.

From Lemma~\ref{L23} and Theorems~\ref{P4} and \ref{T4}, we can derive the evaluation in the large deviation regime.
%(For the relation with \cite{Ogata}, see Section \ref{S8}.)
%%%%
\begin{corollary} %\cite[Theorem~3.1.2]{DZ} 
\Label{theorem:general-markov-ldp}
For arbitrary $\delta > 0$, we have
\begin{gather} 
%\Label{eq:general-markov-ldp-up}
\lim_{n \to \infty} - \frac{1}{n} \log \Pr\left\{ 
X^n - \phi'(0) 
 \ge \delta \right\} %\nonumber \\
= \sup_{\theta \ge 0}[\theta( %\frac{d\phi}{d\theta}
\phi'(0)+\delta) - \phi(\theta) ], \Label{a1}
\\
%\Label{eq:general-markov-ldp-low}
\lim_{n \to \infty} - \frac{1}{n} \log \Pr
\left\{ X^n - \phi'(0) 
 \le - \delta \right\}
%\nonumber \\
 = \sup_{\theta \le 0}[\theta( %\frac{d\phi}{d\theta}
\phi'(0)-\delta) - \phi(\theta)].\Label{a2}
\end{gather}
\end{corollary}

%%%%%%%%%%%%%%%%%%%%%%%%%%%%%%%%%%%%%%%%%%%%%%%%%%%%%%%%
%%%%%%%%%%%%%%%%%% Moderate Deviation %%%%%%%%%%%%%%%%%%%%%%%%%%%%
%\subsection{Moderate Deviation} \Label{\thetaubsection:general-markov-mdp}
\begin{proof}
As mentioned in the end of Section \ref{S2}, 
$\frac{d \phi}{d\theta}$ is a continuously increasing function.
Hence, the RHSs of \eqref{a1} and \eqref{a2} are continuous.

Lemma~\ref{L23} guarantees that 
(RHS of \eqref{27-32})/$n$ goes to 
RHS of \eqref{a1}.
In the RHS of \eqref{27-33}, for given $s > 0 $ and $ \theta > {\phi'}^{-1}(a)$,
the value $e^{-n D({\phi'}^{-1}(a)\|{\theta})
+\overline{\delta}_{\rho}({\phi'}^{-1}(a)) -\underline{\delta}_{\rho}(\theta)
} $ 
and 
$\frac{1}{ns}[\overline{\delta}_{\rho}((1+s)\theta) - \underline{\delta}_{\rho}(\theta)]$
go to zero as $n \to \infty$.
So, (RHS of \eqref{27-33})/$n$ goes to 
RHS of \eqref{a1}.
We can show \eqref{a2} in the same way.
\end{proof}

From Theorems~\ref{P4} and \ref{T4}, we can derive 
the evaluation in the moderate deviation regime.
%%%%%
\begin{corollary} \Label{theorem:general-markov-mdp}
For arbitrary $t \in (0,1/2)$ and $\delta > 0$, we have
\begin{gather} 
\lim_{n \to \infty} - \frac{1}{n^{1-2t}}  \log \Pr\left\{ 
X^n - \phi'(0) 
 \ge n^{-t} \delta  \right\}% \nonumber \\
= \frac{\delta^2}{2 \phi''(0)},
\Label{eq:general-markov-mdp-up} \\
 \lim_{n \to \infty} - \frac{1}{n^{1-2t}}  
\log \Pr
\left\{
X^n - \phi'(0) 
 \le - n^{-t} \delta  \right\} 
%\nonumber \\
= \frac{\delta^2}{2 \phi''(0)}.
\end{gather}
\end{corollary}
%%%
\begin{proof}
Corollary \ref{theorem:general-markov-mdp} can be shown 
by using Lemma~\ref{L23} and Theorems~\ref{P4} and \ref{T4}
in the same way as the proof of Corollary 8.4 of \cite{W-H}.
\end{proof}

\section{Calculation formula of $\phi''(0)$}\label{cal}

When $\Lambda = \sum_\omega C_\omega$ is primitive, 
as a quantum version of Theorem~7.7 of \cite{W-H},
this section gives a useful calculation formula of $\phi''(0)$, which is used for 
the central limit theorem and 
the evaluation in the moderate deviation regime.
In this derivation, Condition (\ref{p-mean}) for the primitivity plays an essential role.

\begin{proposition}
	Let $\tilde{\Lambda}$ be the trace-preserving completely positive map on $\mathcal{T}(\mathcal{H})$ 
	which maps any Hermitian matrix $H$ to $(\Tr H)\rho_0$. 
	Then, the map $Z := (\iota - (\Lambda - \tilde{\Lambda}))^{-1}$ exists and the equations 
	\begin{gather}
		\tilde{\Lambda} = \lim_{n\to\infty} \Lambda^n,\label{eq1}\\
		Z = \sum_{n=0}^\infty (\Lambda - \tilde{\Lambda})^n
		= \iota + \sum_{n=1}^\infty (\Lambda^n - \tilde{\Lambda})\nonumber
	\end{gather}
	hold, where $\iota$ is the identity map on $\mathcal{T}(\mathcal{H})$.
\end{proposition}

The map $Z$ is called the fundamental matrix in the classical case \cite{K-S},
and so $Z$ can be regarded as a quantum version of the fundamental matrix.

\begin{proof}
	The equation (\ref{eq1}) follows from Condition (\ref{p-mean}). 
	From the definition, the equation 
	$\tilde{\Lambda}\circ\Lambda = \Lambda\circ\tilde{\Lambda}
	= \tilde{\Lambda}^2 = \tilde{\Lambda}$ 
	holds. We show that {the magnitude of} any eigenvalue of $\Lambda - \tilde{\Lambda}$ is smaller than one. 
	Let $\lambda\in\mathbb{C}$ be an eigenvalue of $\Lambda - \tilde{\Lambda}$ and 
	$M$ be an eigenvector associated with $\lambda$ of $\Lambda - \tilde{\Lambda}$. 
	Then, any natural number $n$ satisfies $\lambda^n M = (\Lambda - \tilde{\Lambda})^n(M) = (\Lambda^n - \tilde{\Lambda})(M)$. 
	The equation (\ref{eq1}) implies that $\lambda^n M = (\Lambda^n - \tilde{\Lambda})(M) \to 0$. 
	Hence, we obtain $\lambda^n\to0$, which means that the magnitude of $\lambda$ is smaller than one.
	\par
	Since {the magnitude of} any eigenvalue of $\Lambda - \tilde{\Lambda}$ is smaller than one, the map $Z$ exists. 
	Let $Z' := \sum_{n=0}^\infty (\Lambda - \tilde{\Lambda})^n$. Then, the equation 
	$Z' \circ (\iota - (\Lambda - \tilde{\Lambda})) = \iota$ means $Z=Z'$.
\end{proof}

\begin{theorem}
	Define the random variable $X$ as $X(\omega) = x_\omega$ and 
	the map $C_X$ on $\mathcal{T}(\mathcal{H})$ as $C_X := \sum_\omega X(\omega) C_\omega$. 
	Then, 
	\begin{equation}
		\phi''(0) = \mathsf{V}_{\rho_0}[X] + 2\Tr C_X\circ(Z - \tilde{\Lambda})\circ C_X(\rho_0),
	\label{Y1}
	\end{equation}
	where $\mathsf{V}_\rho[X']$ is the variance of a random variable $X'$ under an initial state $\rho$.
\end{theorem}

The symbols $\rho_0$, $\tilde{\Lambda}$, $C_X$ and $Z$ can be calculated 
from the instrument $\{ C_\omega \}_{\omega\in\Omega}$ and the random variable $X$. 
Hence, we can calculate $\phi''(0)$ using the above formula.

\begin{remark}
Here, we explain the relation with
the results by \cite{H-G,GK}.
The paper \cite{H-G} considered the case when the correlation is generated by 
an isometry, and derived the central limit theorem for the sample mean of 
the measurement outcomes.
The combination of the isometry and the measurement can be 
written as an instrument $\{C_\omega\}_{\omega \in \Omega}$.
However, the paper \cite{H-G} addressed multiple variables. 
Hence, when we focus on a single variable,
our model contains the model in \cite{H-G} as a special case.
In addition, the paper \cite{H-G} did not give the relation
between the asymptotic variance and the second derivative of $\phi$.

In contrast, the paper \cite{GK} addressed local asymptotic normality
in the context of system identification for quantum Markov chains.
To discuss local asymptotic normality, 
the paper \cite{GK} calculated the Fisher information matrix.
Hence, it did not discuss the random variable $X^n$.
\end{remark}

\begin{proof}
	Using the equation $\phi''_{n,\rho}(0) = \mathsf{V}_\rho[nX^n]$, Lemma~\ref{moment} in Appendix, and Theorem~\ref{L6}, 
	we have 
	\begin{align*}
		&\lim_{n\to\infty} (1/n)\phi''_{n,\rho}(0)
		= \lim_{n\to\infty} \mathsf{V}_\rho[\sqrt{n} X^n]
		= \lim_{n\to\infty} \mathsf{V}_\rho[\sqrt{n}(X^n - \phi'(0))]\\
		=& \lim_{n\to\infty}\Bigl( \mathsf{E}_\rho\bigl[ \bigl( \sqrt{n}(X^n-\phi'(0)) \bigr)^2 \bigr]
		- \mathsf{E}_\rho[\sqrt{n}(X^n-\phi'(0))]^2 \Bigr)
		= \phi''(0),
	\end{align*}
	where $\mathsf{E}_\rho[X']$ denotes the expectation of a random variable $X'$ under an initial state $\rho$. 
	Put $\rho=\rho_0$. Noting $\Lambda(\rho)=\rho$, we have 
	\begin{align*}
		\mathsf{E}_\rho[X_i] &= \sum_{\omega_1,\ldots,\omega_n} X(\omega_i)\Tr C_{\omega_n}\circ\cdots\circ C_{\omega_1}(\rho)\\
		&= \sum_{\omega_i} X(\omega_i)\Tr \Lambda^{n-i} C_{\omega_i}\Lambda^{i-1}(\rho)
		= \sum_{\omega_i} X(\omega_i)\Tr C_{\omega_i}(\rho) = \mathsf{E}_\rho[X].
	\end{align*}
	Similarly, $\mathsf{E}_\rho[X_i^2] = \mathsf{E}_\rho[X^2]$ holds. 
	Combining these equations, we have 
	\begin{align*}
		&\phi''_{n,\rho}(0) = \mathsf{V}_\rho[nX^n]
		= \mathsf{E}_\rho\Bigl[ \Bigl( \sum_{i=1}^n X_i \Bigr)^2 \Bigr] - \mathsf{E}_\rho\Bigl[ \sum_{i=1}^n X_i \Bigr]^2\\
		=& \sum_{i=1}^n \mathsf{E}_\rho[X_i^2] + 2\sum_{i<j} \mathsf{E}_\rho[X_iX_j] - \mathsf{E}_\rho\Bigl[ \sum_{i=1}^n X_i \Bigr]^2\\
		=& n\mathsf{E}_\rho[X^2] + 2\sum_{i<j} \mathsf{E}_\rho[X_iX_j] - n^2\mathsf{E}_\rho[X]^2\\
		=& n\mathsf{V}_\rho[X] - n(n-1)\mathsf{E}_\rho[X]^2\nonumber\\
		&\quad+ 2\sum_{i<j} \sum_{\omega_i,\omega_j} X(\omega_i) X(\omega_j)
		\Tr\,\Lambda^{n-j}\circ C_{\omega_j}\circ\Lambda^{j-i-1}\circ C_{\omega_i}\circ\Lambda^{i-1}(\rho)\\
		\overset{(a)}{=}& n\mathsf{V}_\rho[X] - n(n-1)\mathsf{E}_\rho[X]^2\\
		&\quad+ 2\sum_{i<j} \sum_{\omega_i,\omega_j} X(\omega_i) X(\omega_j) \Tr C_{\omega_j}\circ\Lambda^{j-i-1}\circ C_{\omega_i}(\rho)\\
		=& n\mathsf{V}_\rho[X] - n(n-1)\mathsf{E}_\rho[X]^2
		+ 2\sum_{i<j} \Tr C_X\circ\Lambda^{j-i-1}\circ C_X(\rho)\\
		=& n\mathsf{V}_\rho[X] - n(n-1)\mathsf{E}_\rho[X]^2
		+ 2\sum_{k=0}^{n-2} (n-k-1)\Tr C_X\circ\Lambda^k\circ C_X(\rho)\\
		=& n\mathsf{V}_\rho[X] + 2\sum_{k=0}^{n-2} (n-k-1)\Tr C_X\circ(\Lambda^k - \tilde{\Lambda})\circ C_X(\rho),
	\end{align*}
	where the equation $\Lambda(\rho)=\rho$ has been used to obtain the equality $(a)$ and 
	the last equality follows from 
	$\Tr\,C_X\circ\tilde{\Lambda}\circ C_X(\rho) = (\Tr\,C_X(\rho))^2 = \mathsf{E}_\rho[X]^2$.
	\par
	Taking the Ces\'aro mean for $Z - \iota = \sum_{n=1}^\infty (\Lambda^n - \tilde{\Lambda})$, we have 
	\begin{equation*}
		Z - \iota = \lim_{n\to\infty} (1/n)\sum_{k=1}^n (n-k+1)(\Lambda^k - \tilde{\Lambda}).
	\end{equation*}
	Hence, we obtain \eqref{Y1} as  
	\begin{align*}
		&\phi''(0) = \lim_{n\to\infty} (1/n)\phi''_{n,\rho}(0)\\
		=& \mathsf{V}_\rho[X]
		+ \lim_{n\to\infty} (2/n)\sum_{k=0}^{n-2} (n-k-1)\Tr C_X\circ(\Lambda^k - \tilde{\Lambda})\circ C_X(\rho)\\
		=& \mathsf{V}_\rho[X] + 2\Tr C_X\circ(Z - \iota)\circ C_X(\rho)
		+ 2\Tr C_X\circ(\iota - \tilde{\Lambda})\circ C_X(\rho)\\
		=& \mathsf{V}_\rho[X] + 2\Tr C_X\circ(Z - \tilde{\Lambda})\circ C_X(\rho).
	\end{align*}
\end{proof}

\section{Relation to finitely correlated states}\Label{S8}
The paper \cite{Ogata} considered finitely correlated states.
Finitely correlated states play a key role in 
analysis on quantum spin chains.
The meaning of large deviation type evaluation in such a physical system is explained in \cite{Ogata,BTG}.

Here, we use a notation for a Hermitian matrix $B$ and a real number $a$: 
$\{B>a\}$ is defined to be the projection $\sum_{j:b_j>a}E_j$ 
when the Hermitian matrix $B$ has the spectral decomposition
$\sum_{j}b_j E_j $.
In the one-dimensional case of this model, we consider a TP-CP map $\Gamma$ from 
${\cal T}({\cal H})$ to
${\cal T}(\C^d)\otimes {\cal T}({\cal H})$.
Then, from an initial state $\rho$ on ${\cal H}$,
we generate the state on $(\C^d)^{\otimes n} $ as
\begin{equation}
\rho_n:= \underbrace{ (\id_{{{\cal T}\left( (\C^d)^{\otimes (n-1)} \right)}} \otimes \Gamma )
\circ \cdots \circ
(\id_{{{\cal T}(\C^d)}} \otimes \Gamma )
\circ \Gamma (\rho)}_{n}.
\end{equation}
Generally, in this model,
we consider the behavior of the average of 
the observables across several systems.
To discuss the relation with our paper, 
we focus only on the case when 
the observables are given as a Hermitian matrix $A$ on the single system 
$\C^d$.

In this special case, 
the paper \cite{Ogata} 
discussed the probability $\Tr \rho_n \{ (1/n)\sum_{i=1}^n A_i > a \}$,
where $A_i$ is defined as $\underbrace{I \otimes \cdots \otimes I}_{i-1}\otimes A \otimes 
\underbrace{I \otimes \cdots \otimes I}_{n-i}$.
That is, it derived the exponential decreasing rate of 
the probability
$\Tr \rho_n \{ (1/n)\sum_{i=1}^n A_i > a \}$
when $n$ goes to infinity.
To convert this model to our model,
we make the spectral decomposition of $A$ as 
$A= \sum_{\omega} x_\omega E_\omega$.
Then, we define the CP map $C_\omega$ on ${\cal H}$ as 
$C_\omega(\rho):= \Tr_{\C^d} (E_\omega \otimes I) \Gamma(\rho)$.
We define the $i$-th variable $X_i$ to be $x_{\omega_i}$ when $\omega_i$ is the $i$-th observation.
Then, we have the relation
\begin{equation}
\Pr\left\{ X^n > a \right\} 
=\Tr \rho_n \Bigl( \Bigl\{ \frac{1}{n}\sum_{i=1}^n A_i > a \Bigr\}\otimes I \Bigr), \Label{LPT}
\end{equation}
which shows that 
our model includes finitely correlated states.

Conversely, when
given a set of CP maps (an instrument) $C=\{C_\omega\}_{\omega \in \Omega}$, 
the map $\sum_{\omega \in \Omega} C_\omega$ is trace-preserving,
we define the TP-CP map $
\Gamma(\rho):= \sum_{\omega \in \Omega} |\omega \rangle \langle \omega| \otimes C_\omega(\rho)$
from 
${\cal T}({\cal H})$ to
${\cal T}(\C^d)\otimes {\cal T}({\cal H})$
and define the Hermitian matrix $A := \sum_\omega x_\omega \ketbra{\omega}{\omega}$, 
where $d$ denotes the number of elements of $\Omega$.
Then, the relation \eqref{LPT} holds.
Therefore, we can conclude that our model is equivalent to the model of finitely correlated states
when the observable is given as a Hermitian matrix $A$ on the single system 
$\C^d$.

When we consider the behavior of the average of 
the observables across more than one systems,
we need to care about the non-commutativity between several observables in general.
This type analysis is more difficult than that in this paper, but 
if we obtain analysis similar to Theorem~\ref{lemma:mgf-finite-evaluation}, 
we can derive analysis similar to Theorems~3--5 by using the same discussion.
Therefore, Corollary 4 can be regarded as a special case of 
\cite[Theorem 1.2 \& (27)]{Ogata}.
However, in this case,
our analysis has the following two advantages over the analysis of the paper 
\cite{Ogata}.
Indeed, the paper \cite{Ogata} did not give any non-asymptotic evaluation. 
On the other hand, we have given the non-asymptotic evaluation of the probability \eqref{LPT} as Theorems~3 and 4,
and have shown their asymptotic tightness.
In this sense, the analysis of this paper is more advanced 
in this special case of finitely correlated states than that in \cite{Ogata}.
The contribution of this paper is the derivation of the non-asymptotic evaluation to achieve the asymptotic tightness.

\section{Discussion}
This paper has addressed a hidden Markovian process with quantum hidden system,
and has derived several asymptotic behaviors of the sample mean, namely 
the central limit theorem, the large and moderate deviations for the tail probability
while no existing paper treated the non-asymptotic behaviors of such a hidden Markovian process with quantum hidden system.
Using the quantum version of the Perron-Frobenius theorem, we have derived upper and lower bounds of the cumulant generating function.
Based on these bounds, we have obtained 
the above result by using the same method as \cite{W-H}.

In particular, Corollary \ref{Co1} can be regarded as the central limit theorem for 
the hidden Markov process with quantum hidden system.
In the classical case, as the refinement of Corollary \ref{Co1},
the paper \cite[Theorem~2]{Herve} showed 
the Markov version of the Berry-Esseen Theorem.
So, it is interesting problem to derive the Berry-Esseen Theorem 
for the hidden Markov process with quantum hidden system.

Further, our method relies only on the quantum version of the Perron-Frobenius theorem, 
and the quantum version of the Perron-Frobenius theorem can be extended to 
a finite-dimensional vector space with a general closed convex cone, in which 
the dual convex cone need not be equal to the original convex cone \cite{Yoshida,Vander,Barker,Barker-Schneider}. 
Therefore, our result can be extended to a more general setting, 
e.g., general probabilistic theory \cite{BBLW,BBLW2,BHK,Gudder1,Gudder2,KMI}.

When the system is given as a hidden system
$X_1, \ldots,X_n  $,
the memory effect does not vanish,
i.e., 
$P_{X_i|X_{i-1},\ldots X_{i-k}}(x_i|x_{i-1},\ldots x_{i-k})$
cannot be simplified to 
$P_{X_i|X_{i-1},\ldots X_{i-k'}}(x_i|x_{i-1},\ldots x_{i-k'})$
with $k'<k$ \cite{hidden}.
Hence, the outcome has long-period memory.
When discussing some information theoretic problem,
we need to discuss information theoretical quantity, e.g., entropy and conditional entropy instead of the sample mean \cite{Markov-HW}.
In this case, we need to be careful with such a memory.
Hence, it is not sufficient to discuss the sample mean of a random variable for this purpose.
For such an application, we need more complicated calculation.

\section*{Appendix}

For readers' convenience, we prove a few well-known lemmas used in Section~\ref{cal}. 
The textbook \cite[Corollary in p.~338]{Billingsley} has a more general statement than Lemma~\ref{moment}.
%Although we use the terms of a Borel function and a Borel set, 
%readers who do not know them may regard them as a usual function and a usual set, respectively.

\begin{lemma}\label{bounded}
	Let $\mu_n$ and $\mu$ be probability distributions on $\mathbb{R}$, 
	and $\mathcal{C}$ be the set of all continuous points of the cumulative distribution function of $\mu$. 
	Assume $\mu_n\to\mu$. 
	Then, any continuous function $f: \mathbb{R}\to\mathbb{R}$ and 
	any two real numbers $a,b\in\mathcal{C}$ with $a<b$ satisfy 
	\begin{equation*}
		\lim_{n\to\infty} \int_{(a,b]} f(x)\,\mu_n(dx) = \int_{(a,b]} f(x)\,\mu(dx).
	\end{equation*}
	Further, any continuous function $f: \mathbb{R}\to\mathbb{R}$ with $\sup_{x>0} |f(x)| < \infty$ and 
	any real number $a\in\mathcal{C}$ satisfy 
	\begin{equation*}
		\lim_{n\to\infty} \int_{(a,\infty)} f(x)\,\mu_n(dx) = \int_{(a,\infty)} f(x)\,\mu(dx).
	\end{equation*}
	Here, $\mu_n\to\mu$ means that 
	$\lim_{n\to\infty} \int_{\mathbb{R}} f(x)\,\mu_n(dx) = \int_{\mathbb{R}} f(x)\,\mu(dx)$ 
	for any bounded continuous function $f: \mathbb{R}\to\mathbb{R}$.
\end{lemma}
\begin{proof}
	Assume $\mu_n\to\mu$. 
	Define the cumulative distribution functions $F_n$ and $F$ of $\mu_n$ and $\mu$ 
	as $F_n(x) := \mu_n((-\infty,x])$ and $F(x) := \mu((-\infty,x])$, respectively. 
	First, we prove that any point $a\in\mathcal{C}$ satisfies $\lim_{n\to\infty} F_n(a) = F(a)$. 
	Take an arbitrary positive number $\epsilon$ and 
	define the two bounded continuous functions $h_\pm$ as 
	\[
	h_+(x) := 
	\begin{cases}
		1 & x\le a,\\
		(a+\epsilon-x)/\epsilon & a<x<a+\epsilon,\\
		0 & x\ge a+\epsilon,
	\end{cases}
	\quad
	h_-(x) := 
	\begin{cases}
		1 & x\le a-\epsilon,\\
		(a-x)/\epsilon & a-\epsilon<x<a,\\
		0 & x\ge a.
	\end{cases}
	\]
	Then, since the inequalities $F_n(a) \le \int_{\mathbb{R}} h_+(x)\,\mu_n(dx)$ and 
	$F_n(a) \ge \int_{\mathbb{R}} h_-(x)\,\mu_n(dx)$ hold, by using the assumption $\mu_n\to\mu$, we obtain 
	\begin{gather*}
		\limsup_{n\to\infty} F_n(a)
		\le \limsup_{n\to\infty} \int_{\mathbb{R}} h_+(x)\,\mu_n(dx)
		= \int_{\mathbb{R}} h_+(x)\,\mu(dx)
		\le F(a+\epsilon),\\
		\liminf_{n\to\infty} F_n(a)
		\ge \liminf_{n\to\infty} \int_{\mathbb{R}} h_-(x)\,\mu_n(dx)
		= \int_{\mathbb{R}} h_-(x)\,\mu(dx)
		\ge F(a-\epsilon).
	\end{gather*}
	By taking the limit $\epsilon\downarrow0$, these inequalities turn out 
	$\limsup_{n\to\infty} F_n(a)\le F(a)$ and $\liminf_{n\to\infty} F_n(a)\ge F(a)$ 
	because of $a\in\mathcal{C}$. 
	That is, the desired equation $\lim_{n\to\infty} F_n(a) = F(a)$ holds.
	\par
	Next, we prove the first equation in this lemma. 
	Let $f: \mathbb{R}\to\mathbb{R}$ be a bounded continuous function and 
	$a,b\in\mathcal{C}$ be two real numbers satisfying $a<b$. 
	Then, we define the function $h: \mathbb{R}\to\mathbb{R}$ as 
	\[
	h(x) := 
	\begin{cases}
		f(a) & x\le a,\\
		f(x) & a<x<b,\\
		f(b) & x\ge b.
	\end{cases}
	\]
	Since the assumption $\mu_n\to\mu$ and the above proof imply that 
	$\lim_{n\to\infty} \int_{\mathbb{R}} h(x)\,\mu_n(dx) = \int_{\mathbb{R}} h(x)\,\mu(dx)$ and 
	$\lim_{n\to\infty} F_n(x) = F(x)$ with $x=a,b$, respectively, we obtain 
	\begin{align*}
		&\lim_{n\to\infty} \int_{(a,b]} f(x)\,\mu_n(dx) + f(a)F(a) + f(b)(1-F(b))\\
		=& \lim_{n\to\infty} \Bigl[ \int_{(a,b]} f(x)\,\mu_n(dx) + f(a)F_n(a) + f(b)(1-F_n(b)) \Bigr]\\
		=& \lim_{n\to\infty} \int_{\mathbb{R}} h(x)\,\mu_n(dx)
		= \int_{\mathbb{R}} h(x)\,\mu(dx)\\
		=& \int_{(a,b]} f(x)\,\mu(dx) + f(a)F(a) + f(b)(1-F(b)).
	\end{align*}
	Therefore, $\lim_{n\to\infty} \int_{(a,b]} f(x)\,\mu_n(dx) = \int_{(a,b]} f(x)\,\mu(dx)$ holds. 
	Similarly, the second equation in this lemma can also be shown.
\end{proof}

\begin{lemma}[Covergence of moments]\label{moment}
	Let $\mu_n$ and $\mu$ be probability distributions on $\mathbb{R}$. 
	If the cumulant generating functions of $\mu_n$ and $\mu$ exist 
	and the cumulant generating functions of $\mu_n$ converge pointwise to that of $\mu$, 
	then any natural number $k$ satisfies 
	\begin{equation*}
		\lim_{n\to\infty} \int_\mathbb{R} x^k\,\mu_n(dx) = \int_\mathbb{R} x^k\,\mu(dx).
	\end{equation*}
\end{lemma}
\begin{proof}
	Assume the cumulant generating functions of $\mu_n$ converge that of $\mu$, which implies $\mu_n\to\mu$. 
	Let $\theta$ be an arbitrary positive number and take a continuous point $a\in\mathcal{C}$. 
	It is sufficient to prove 
	$\lim_{n\to\infty} \int_\mathbb{R} (x-a)^k\,\mu_n(dx) = \int_\mathbb{R} (x-a)^k\,\mu(dx)$ 
	for any natural number $k$. 
	Since $e^{\theta(x-a)}$ is a continuous function, 
	Lemma~\ref{bounded} implies 
	\begin{equation*}
		\lim_{n\to\infty} \int_{(-\infty,a]} e^{\theta(x-a)}\,\mu_n(dx) = \int_{(-\infty,a]} e^{\theta(x-a)}\,\mu(dx).
	\end{equation*}
	This equation and the convergence of the cumulant generating functions imply 
	\begin{equation*}
		\lim_{n\to\infty} \int_{(a,\infty)} e^{\theta(x-a)}\,\mu_n(dx) = \int_{(a,\infty)} e^{\theta(x-a)}\,\mu(dx).
	\end{equation*}
	Combining the above equation and the inequality $e^{\theta(x-a)} \ge \theta(x-a) + 1$, we have 
	\begin{align*}
		\limsup_{n\to\infty} \int_{(a,\infty)} (x-a)^k\,\mu_n(dx)
		&\le \limsup_{n\to\infty} \int_{(a,\infty)} \Bigl( \frac{e^{\theta(x-a)} - 1}{\theta} \Bigr)^k\,\mu_n(dx)\\
		&= \int_{(a,\infty)} \Bigl( \frac{e^{\theta(x-a)} - 1}{\theta} \Bigr)^k\,\mu(dx).
	\end{align*}
	Since $(e^{\theta(x-a)} - 1)/\theta$ monotonically increases with respect to $\theta$, 
	the monotone convergence theorem yields 
	\begin{equation*}
		\int_{(a,\infty)} \Bigl( \frac{e^{\theta(x-a)} - 1}{\theta} \Bigr)^k\,\mu(dx)
		\xrightarrow{\theta\downarrow0} \int_{(a,\infty)} (x-a)^k\,\mu(dx).
	\end{equation*}
	Hence, 
	\begin{equation*}
		\limsup_{n\to\infty} \int_{(a,\infty)} (x-a)^k\,\mu_n(dx) \le \int_{(a,\infty)} (x-a)^k\,\mu(dx).
	\end{equation*}
	\par
	Take an arbitrary real number $b\in\mathcal{C}$ satisfying $a<b$. 
	Since $(x-a)^k$ is a continuous function, 
	Lemma~\ref{bounded} implies 
	\[
	\liminf_{n\to\infty} \int_{(a,\infty)} (x-a)^k\,\mu_n(dx)
	\ge \liminf_{n\to\infty} \int_{(a,b]} (x-a)^k\,\mu_n(dx)
	= \int_{(a,b]} (x-a)^k\,\mu(dx).
	\]
	Note that $\mathbb{R}\setminus\mathcal{C}$ is a countable set because the cumulative distribution function $F$ monotonically increases. 
	Thus, we can take the limit $b\to\infty$ while satisfying $b\in\mathcal{C}$. 
	Taking this limit, we have 
	\begin{equation*}
		\liminf_{n\to\infty} \int_{(a,\infty)} (x-a)^k\,\mu_n(dx) \ge \int_{(a,\infty)} (x-a)^k\,\mu(dx).
	\end{equation*}
	Thus, the equation $\lim_{n\to\infty} \int_{(a,\infty)} (x-a)^k\,\mu_n(dx) = \int_{(a,\infty)} (x-a)^k\,\mu(dx)$ holds. 
	The other equation $\lim_{n\to\infty} \int_{(-\infty,a]} (x-a)^k\,\mu_n(dx) = \int_{(-\infty,a]} (x-a)^k\,\mu(dx)$ 
	can be also shown in a similar way. 
	From these equations, we obtain the desired equation.
\end{proof}

\color{black}
\section*{Acknowledgments}
MH is very grateful to 
%Professor Vincent Y. F. Tan and 
Professor Shun Watanabe for helpful discussions and comments.
He was supported in part by a JSPS Grant-in-Aids for Scientific Research (B) No.16KT0017
and for Scientific Research (A) No.17H01280, 
the Okawa Research Grant and Kayamori Foundation of Information Science Advancement.

\end{document}